\documentclass[conference]{IEEEtran}
\IEEEoverridecommandlockouts 
\ifCLASSINFOpdf
\usepackage[pdftex]{graphicx}
\graphicspath{{./figure/ }}
\DeclareGraphicsExtensions{.pdf,.jpeg,.png}
\else
\usepackage[dvips]{graphicx}
\DeclareGraphicsExtensions{.eps}
\fi

\usepackage{tikz}
\usepackage{pgfplots}
\pgfplotsset{compat=1.15}
\usetikzlibrary{calc,shapes.geometric,shapes.misc,arrows,matrix}
\usepackage[absolute,overlay]{textpos}

\usepackage{subfigure}
\graphicspath{{figure/}}

\usepackage[utf8]{inputenc}
\usepackage[nolist]{acronym}
\usepackage[cmex10]{amsmath}
\usepackage{mdwmath}
\usepackage{mdwtab}
\usepackage{subfigure}
\graphicspath{{./figure/}}
\usepackage{afterpage}
\usepackage{multirow}
\usepackage{amsthm}
\usepackage{amssymb}

\newtheorem{proposition}{Proposition}

\usepackage{siunitx}
\DeclareSIUnit{\dBi}{dBi}
\DeclareSIUnit{\dBm}{dBm}
\DeclareSIUnit{\dBW}{dBW}
\usepackage{color}
\usepackage{dashrule}
\usepackage{algorithm} 
\usepackage{algpseudocode}
\usepackage{cite}
\usepackage{float}
\floatstyle{plaintop}
\restylefloat{table}
\usepackage{balance}

\usepackage{amssymb}

\newcommand{\vek}[1]{\ensuremath{\mathbf{#1}}}          
\newcommand{\vekt}[1]{\ensuremath{\tilde{\vek{#1}}}}        

\newcommand{\Exvl}[1]{\ensuremath{\mathrm{E}\left\{#1\right\}}}
\newcommand{\var}[1]{\ensuremath{\mathrm{Var}\left\{#1\right\}}}
\newcommand{\nvar}{\ensuremath{\sigma_\mathsf{n}^2}}
\newcommand{\naltvar}{\bar{\sigma}_{\mathsf{n}}^2}

\newcommand{\NS}{\ensuremath{N_{\text{S}}}}          

\newcommand{\Nt}{\ensuremath{N_{\text{t}}}}
\newcommand{\Nr}{\ensuremath{N_{\text{r}}}}
\newcommand{\NTx}{\ensuremath{N_{\text{Tx}}}}

\newcommand{\GTx}{\ensuremath{\zeta_{\text{Tx,dB}}}}
\newcommand{\GRx}{\ensuremath{\zeta_{\text{Rx,dB}}}}

\newcommand{\DA}{\ensuremath{D_{\text{A}}}}
\newcommand{\DS}{\ensuremath{D_{\text{S}}}}

\newcommand{\Ptx}{\ensuremath{P_{\text{Tx}}}}

\newcommand{\aoa}{\ensuremath{\theta}}
\newcommand{\aoal}{\ensuremath{\theta_\ell}}

\newcommand{\aodl}{\ensuremath{\Theta_\ell}}

\newcommand{\RE}{\ensuremath{r_{\text{E}}}}

\newcommand{\C}{\ensuremath{\mathbb{C}}}

\newcommand{\fc}{\ensuremath{f_{\text{c}}}}

\newcommand{\veka}{\ensuremath{\vek{a}}}
\newcommand{\vekb}{\ensuremath{\vek{b}}}
\newcommand{\vekx}{\ensuremath{\vek{x}}}
\newcommand{\veks}{\ensuremath{\vek{s}}}
\newcommand{\veksest}{\ensuremath{\vekt{s}}}

\newcommand{\veky}{\ensuremath{\vek{y}}}

\newcommand{\vekn}{\ensuremath{\vek{n}}}

\newcommand{\EQ}{\ensuremath{\vek{W}}}
\newcommand{\eq}{\ensuremath{\vek{w}}}
\newcommand{\eqsat}{\ensuremath{\vek{w}_{\ell}}}
\newcommand{\PC}{\ensuremath{\vek{G}}}
\newcommand{\pc}{\ensuremath{\vek{g}}}

\newcommand{\vekH}{\ensuremath{\vek{H}}}

\newcommand{\vekHsat}{\ensuremath{\vek{H}_{\ell}}}

\newcommand{\chtx}[1]{\ensuremath{\vekH_{#1}}}

\newcommand{\chest}{\ensuremath{\vekt{H}}}

\newcommand{\vekI}{\ensuremath{\vek{I}}}

\newcommand{\vekS}{\ensuremath{\vek{S}}}

\newcommand{\vekr}{\ensuremath{\vek{r}}}
\newcommand{\vekd}{\ensuremath{\vek{d}}}
\newcommand{\vekP}{\ensuremath{\vek{P}}}

\newcommand{\vekU}{\ensuremath{\vek{U}}}
\newcommand{\vekV}{\ensuremath{\vek{V}}}
\newcommand{\vekSig}{\ensuremath{\boldsymbol{\Sigma}}}
\newcommand{\vekA}{\ensuremath{\vek{A}}}
\newcommand{\vekB}{\ensuremath{\vek{B}}}

\newcommand{\dB}{\ensuremath{\text{dB}}}

\newcommand{\diag}[1]{\ensuremath{\operatorname{diag}\left(#1\right)}}
\newcommand{\blkdiag}[1]{\ensuremath{\operatorname{blkdiag}\left(#1\right)}}
\newcommand{\tr}[1]{\ensuremath{\operatorname{tr}\left\{#1\right\}}}

\begin{acronym}
	\acro{rx}[GS]{ground station}
	\acro{dl}[DL]{downlink}
	\acro{mimo}[MIMO]{multiple-input-multiple-output}
	\acro{ntn}[NTN]{non-terrestrial network}
	\acro{geo}[GEO]{geostationary orbit}
	\acro{meo}[MEO]{medium Earth orbit}
	\acro{leo}[LEO]{low Earth orbit}
	\acro{qos}[QoS]{quality of service}
	\acro{dip}[DiP]{distributed precoding}
	\acro{pspc}[PAPPC]{per-acces point power constraint}
	\acro{kkt}[KKT]{Karush-Kuhn-Tucker}
	\acro{svd}[SVD]{singular value decomposition}
	\acro{mse}[MSE]{mean-squared error}
	\acro{mmse}[MMSE]{minimum mean-squared error}
	\acro{zf}[ZF]{zero-forcing}
	\acro{mrt}[MRT]{maximum ratio transmission}
	\acro{mrc}[MRC]{maximum ratio combining}
	\acro{snr}[SNR]{signal-to-noise ratio}
	\acro{sinr}[SINR]{signal-to-interference-and-noise ratio}
	\acro{csi}[CSI]{channel state information}
	\acro{csit}[CSIT]{\ac{csi} at the transmitter}
	\acro{csir}[CSIR]{\ac{csi} at the receiver}
	\acro{los}[LOS]{line-of-sight}
	\acro{aod}[AoD]{angle of departure}
	\acro{aoa}[AoA]{angle of arrival}
	\acro{ula}[ULA]{uniform linear array}
	\acro{ecef}[ECEF]{Earth-centered, Earth-fixed }
	\acro{dft}[DFT]{discrete Fourier transform}
	\acro{sota}[SotA]{state of the art}
\end{acronym}

\allowdisplaybreaks

\usepackage{moreverb}
\immediate\write18{texcount -inc -incbib 
-sum main.tex > /tmp/wordcount.tex}

\begin{document}	
\title{Beamspace MIMO for Satellite Swarms}

\author{
\IEEEauthorblockN{
Maik Röper\IEEEauthorrefmark{1}, Bho Matthiesen\textsuperscript{\IEEEauthorrefmark{1}\IEEEauthorrefmark{2}}, Dirk Wübben\IEEEauthorrefmark{1}, Petar Popovski\textsuperscript{\IEEEauthorrefmark{3},\IEEEauthorrefmark{2}} and Armin Dekorsy\IEEEauthorrefmark{1}}%
\IEEEauthorblockA{\IEEEauthorrefmark{1} Gauss-Olbers Center, c/o University of Bremen, Dept. of Communications Engineering, 28359 Bremen, Germany\\
	\IEEEauthorrefmark{2} University of Bremen, U Bremen Excellence Chair, Dept.\ of Communications Engineering, 28359 Bremen, Germany\\
\IEEEauthorrefmark{3} Aalborg University, Department of Electronic Systems, 9220 Aalborg, Denmark\\
	Email: \{roeper,matthiesen,wuebben,dekorsy\}@ant.uni-bremen.de, petarp@es.aau.dk}
}

\maketitle

\begin{abstract}
	Systems of small distributed satellites in low Earth orbit (LEO) transmitting cooperatively to a multiple antenna ground station (GS) are investigated. These satellite swarms have the benefit of much higher spatial separation in the transmit antennas than traditional big satellites with antenna arrays, promising a massive increase in spectral efficiency. However, this would require instantaneous perfect channel state information (CSI) and strong cooperation between satellites. In practice, orbital velocities around 7.5\,km/s lead to very short channel coherence times on the order of  fractions of the inter-satellite propagation delay, invalidating these assumptions. In this paper, we propose a distributed linear precoding scheme and a GS equalizer relying on local position information. In particular, each satellite only requires information about its own position and that of the GS, while the GS has complete positional information. Due to the deterministic nature of satellite movement this information is easily obtained and no inter-satellite information exchange is required during transmission. Based on the underlying geometrical channel approximation, the optimal inter-satellite distance is obtained analytically. Numerical evaluations show that the proposed scheme is, on average, within 99.8\,\% of the maximum achievable rate for instantaneous CSI and perfect cooperation
\end{abstract}


\section{Introduction}
Integrating \acp{ntn} into terrestrial communication systems is an important step towards truly ubiquitous connectivity \cite{3GPPTR22.822,Kodheli.etal.2021}. An essential building block are small satellites in \ac{leo} that are currently deployed in private sector mega constellations \cite{Portillo.Cameron.Crawley.2019,Di.Song.Li.Poor.2019,Leyva-Mayorga2020}. Their main benefits are much lower propagation delays and deployment costs due to the \ac{leo} when compared to more traditional high-throughput satellites \cite{Zheng.Chatzinotas.Ottersten.2012,Joroughi.Vazquez.Perez-Neira.2016,Perez-Neira.Vazquez.Shankar.Malekli-Chatzinotas.2019} in \ac{meo} and \ac{geo}. While current systems focus on connecting \acp{rx} to a single satellite, combining several low cost satellites in swarms leads to increased flexibility and scalability \cite{Verhoeven.Bentum.Monna.Rotteveel.Guo.2011}.

Especially the joint transmission of multiple satellites forming large virtual antenna arrays promises tremendous spectral efficiency gains solely due to the increased spatial separation of antennas \cite{Budianu.Meijernik.Bentum2015,Richter.Bergel.Noam.Yair.2020,Schwarz.Delamotte.Storek.Knoop2019}. However, the straightforward implementation of conventional \ac{mimo} transmission schemes requires complete instantaneous \ac{csi} and inter-satellite coordination of joint beamforming. This is infeasible due to very short channel coherence times resulting from high orbital velocities in combination with comparably large propagation delays, both in ground-to-satellite and in inter-satellite links. In this paper, we show that this is not an obstacle if positional information is exploited. In contrast to complete \ac{csi}, this information is often readily available or easily determined from the deterministic movement of satellites. This leads to an approximate channel model, 
which is employed to derive a beamspace \ac{mimo} \cite{Lin.Gao.Jin.Li.2017,Ahmed2018} based distributed linear precoder and equalizer. The precoder has low complexity, requires, at each satellite, only knowledge of the own rotation and the position of itself and the \ac{rx}, and achieves close to optimal spectral efficiency. Similarly, the equalizer only needs \ac{aoa} information for the satellites and, given proper design of the satellite swarm, shows nearly optimal performance. We obtain an analytical solution on the optimal swarm layout and numerically evaluate the system performance.

The related literature can be summarized as follows:
In \cite{Budianu.Meijernik.Bentum2015}, the \ac{dl} from a satellite swarm with more than 50 nano-satellites towards a single antenna \acf{rx} is studied. It is shown that, if the signals of all satellites add up in phase at the \ac{rx} a high array gain is achieved.
Communication between multiple satellites and a \ac{rx} with multiple antennas is studied in \cite{Yamashita.Kobayashi.Ueba.Umehira.2005}, where an iterative interference cancellation algorithm is considered to deal with the large spatial correlation between two close \ac{geo} satellites.
Furthermore, in \cite{Goto.Shibayama.Yamashita.Yamazato.2018} and \cite{Liolis.Panagopoulos.Cottis.2007}, the capacity of multi-satellite systems are studied.
In \cite{Roeper.Dekorsy.2019}, a distributed precoding algorithm, based on the \ac{mmse} criterion and exploiting information exchange between the satellites, is proposed for a multi-user \ac{dl} scenario.
In \cite{Richter.Bergel.Noam.Yair.2020} a \ac{zf} equalizer at the ground terminal is proposed while receiving from two satellites.
In \cite{You.Li.Wang.Gao.Xia.Ottersten.2020,Guo.Lu.Gao.Xia.2021,Lin.Lin.Champagne.Zhu.AlDhahir.2020}, beamspace \ac{mimo} is adapted to ground to satellite communications, focusing on scenarios involving a single satellite.

\section{System Model and Performance Bounds}\label{sec:system}
Consider a swarm of $\NS$ satellites flying in a trail formation with constant inter-satellite distance $\DS$. They have a common circular orbit at an altitude of $d_0$ that is assumed to be ideal Keplerian and aligned within the $\mathsf{xy}$-plane.
Then, the polar coordinates of satellite $\ell$ in the Earth-centered coordinate system are denoted by $\vekr_\ell = [r_0, \vartheta_\ell]^T$, where $r_0$  is the orbital radius, i.e., the distance from the center of the Earth to the satellite, and $\vartheta_\ell$ is the polar angle. 
Given the Earth's radius $\RE=\SI{6371}{\kilo\meter}$, the orbital radius is $r_0=\RE + d_0$. The \ac{rx} is located within the orbital plane at position $\vekr_\text{Rx} = [\RE, \pi/2]$. This setup is illustrated in Fig.~\ref{fig:angles}.

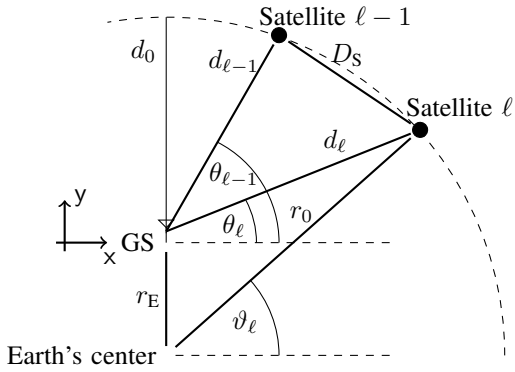
\begin{figure}
	\centering
	\tikzstyle{block} = [draw, fill=white, rectangle, 
    minimum height=3em, minimum width=3.3em]

\newcommand{\Nsat}{N_{\mathrm{S}}}

\tikzset{sat/.pic ={
	\draw[fill=black] (0,0) circle [radius=0.1];
}}

\tikzset{rxantenna/.pic ={
		\draw (0,0) -- (0,0.2) -- (-0.1,0.3) -- (0.1,0.3) -- (0,0.2);
}}


\definecolor{darkgreen}{rgb}{0.12549019607843137255,0.4980392156862745098,0.16862745098039215686}

\begin{tikzpicture}
 	[scale=1.5,node/.style={circle,draw,fill=white,circular drop shadow,thick,inner sep=0pt,minimum size=10mm,fill=white},
 	arrow/.style={->,shorten <=1pt,shorten >=1pt,thick},
 	branch/.style={circle,fill, minimum size=6pt,inner sep=0pt},
 	arrow_d/.style={->,shorten <=1pt,shorten >=1pt,>=stealth',semithick,dashed}]

\draw[arrow] (-2.5,0) --  node [below, at end] {$\mathsf{x}$} (-2,0);
\draw[arrow] (-2.4,-0.1) --  node [right, at end] {$\mathsf{y}$} (-2.4,0.4);

\node (sat1) at (-0.5,1.84) {};
\node (satNs) at (0.75,1) {};
\node (rx) at (-1.5,0) {};
\node (center) at (-1.5,-1) {};

\node () at (0,2)  {Satellite $\ell-1$};
\node () at (1.1,1.2)  {Satellite $\ell$};
\node () at (-1.75,0)  {GS};
\node () at (-2.25,-1)  {Earth's center};

\draw (rx) -- (-1.5,2);
\node (d0) at (-1.7,1.7) {$d_0$};

\pic[rotate=-10] (Sat1) at (sat1) {sat};
\draw[thick] (sat1) -- (-1.5,0.1);
\node (d1) at (-0.9,1.6) {$d_{\ell-1}$};	
%
%
\pic[rotate=-80] (SatNs) at (satNs) {sat};
\draw[thick] (satNs) -- (-1.5,0.1);
\node (dNs) at (0,0.9) {$d_{\ell}$};

\draw[thick] (satNs) -- (sat1);
\node () at (0.1,1.66) {$D_\text{S}$};

\pic (rNr) at (rx) {rxantenna};
	
\draw[dashed] (rx) -- (0.5,0);
\draw (rx) ++ (0:0.8) arc (0:28:0.8);
\node () at (-0.9,0.15) {$\theta_\ell$};
\draw (-1.5,0) ++ (0:1) arc (0:63:1);
\node () at (-0.9,0.6) {$\theta_{\ell-1}$};

\draw[dashed] (center) ++ (0:3) arc (0:100:3);

\draw[thick] (center) -- (satNs);
\node () at (-0.3,0.25) {$r_0$};

\draw[thick] (center) -- (rx);
\node () at (-1.66,-0.5) {$r_\text{E}$};

\draw[dashed] (center) -- (0.5,-1);
\draw (center) ++ (0:1) arc (0:42:1);
\node () at (-0.8,-0.7) {$\vartheta_\ell$};

%

	

\end{tikzpicture}%
	\caption{Geometric relation between Satellites and \ac{rx}}
	\label{fig:angles}%
\end{figure}

It is convenient to describe the satellites position in a $\ac{rx}$ centered coordinate system. The relative polar coordinates of satellite $\ell$ to the \ac{rx} are denoted as $\vekd_\ell = [d_\ell, \aoal]$, where $\aoal\in[0,\pi]$ is the elevation angle, i.e., the polar angle in the \ac{rx} centered coordinate frame. This angle is equivalent to the \ac{aoa} of the signal from satellite $\ell$ at the \ac{rx}. Considering the triangle between satellite $\ell$, the \ac{rx} and the Earth's center, we obtain from the law of sines $\vartheta_\ell = \aoal + \arcsin(\RE \cos(\aoal) / r_0 )$. Correspondingly, the distance $d_\ell$ between satellite $\ell$ and the \ac{rx} is $d_\ell = \sqrt{d_0^2 + 2\RE r_0\left( 1 -  \sin\left(\vartheta_\ell\right)\right)}$.
The \ac{aod} $\aodl$ from satellite $\ell$ is readily obtained from the elevation angle $\aoal$ and the satellite's rotation $\eta_\ell$ as
$\aodl = \aoal - \eta_\ell - \frac{\pi}{2}$, where $\eta_\ell$ is defined such that the antenna arrays at the satellite and the \ac{rx} are parallel to each other at $\eta_\ell = 0$. Hence, if the satellite points perfectly towards the \ac{rx}, the \ac{aod} is $\aodl=0$. 
Assume that the satellite can transmit only in the directions $\aodl\in[-\pi/2,\pi/2]$. Then, its rotation must be within the interval $\eta_\ell\in[\min(-\pi/2,\aoal - \pi), \min(\pi/2, \aoal)]$.

Observe that $r_0$ and $\RE$ are constant over time, while $d_\ell$, $\vartheta_\ell$, $\aoal$, $\eta_\ell$, and $\aodl$ are time variant. All of these values are either known a~priori at the satellites and the \ac{rx} or can be obtained easily because the satellites are moving on predefined orbits.

\subsection{Communication Model} \label{sec:system_trans}
Satellite $\ell$ is equipped with an \ac{ula} consisting of $\Nt$ antennas spaced $\DA=\frac{c_0}{2 \fc}$ apart, where $\fc$ is the carrier frequency and $c_0$ is the speed of light. The satellites jointly transmit $M$ independent messages that are known a~priori at all satellites and encoded in $\veks\in\C^M$ using uncorrelated unit variance Gaussian codebooks. Satellite $\ell$ employs linear precoding to transmit the signal $\vekx_\ell = \PC_\ell\veks$ with $\PC_\ell\in\C^{\Nt\times M}$. The transmission is subject to an average power constraint $\rho_\ell$, i.e.,
\begin{align}\label{eq:power_constr}
	\tr{\PC_\ell\PC_\ell^H} \leq \rho_\ell\,.
\end{align}

The \ac{rx} is equipped with an \ac{ula} consisting of $\Nr\ge \NS$ antenna elements.
Its received signal is
$\veky = \sum_{\ell=1}^{\NS} \vekHsat\vekx_{\ell} + \vekn$,
where $\vekn$ is independent and identically distributed (i.i.d.) complex circularly symmetric white Gaussian noise with power $\sigma_\mathsf{n}^2$ and $\vekHsat\in\C^{\Nr\times\Nt}$ is the channel from satellite $\ell$ to the \ac{rx}. Due to the collaborative transmission, this is effectively a point-to-point channel. In particular, let $\vekH=\left[\vekH_1, ... ,\vekH_{\NS} \right]$ and $\vekx=\left[\vekx_1^T, ... ,\vekx_{\NS}^T \right]^T$ to obtain the equivalent channel
$$\veky = \vekH \vekx + \vekn.$$
The capacity of this channel is
\begin{align}\label{eq:rate_p2p}
	R_{\text{opt}} 
	&= \max_{\tr{\PC\PC^H} \leq \sum_{\ell=1}^{\NS} \rho_\ell} \log_2\left\vert \vekI_{\Nr} + \sigma_\mathsf{n}^{-2} \vekH\PC\PC^H\vekH^H \right\vert,
\end{align}
where $\PC=[\PC_1^T,...,\PC_{\NS}^T]^T\in\C^{\NTx\times M}$ with $\NTx = \NS \Nt$ \cite{Telatar.1999}. The maximum in \eqref{eq:rate_p2p} is achieved for
\begin{align}\label{eq:pc_opt}
	\PC_\text{opt} = \vekV\vekP^{\frac{1}{2}},
\end{align}
where $\vekU\vekSig\vekV^H$ is the \ac{svd} of $\vekH$, and $\vekP=\diag{p_1,...,p_{\NTx}}$ with $p_\mu$ the transmit power of the $\mu$th beam. Combining \eqref{eq:rate_p2p} and \eqref{eq:pc_opt}, we obtain 
\begin{align}\label{eq:rate_opt}
	R_{\text{opt}} = \sum_{\mu=1}^{M} \log_2\left(1 + \lambda_\mu\frac{p_\mu}{\nvar} \right),
\end{align}
where $\lambda_\mu$ the $\mu$th Eigenvalue of $\vekH\vekH^H$.
The optimal power allocation $\vekP$ is obtained from the waterfilling algorithm \cite{Telatar.1999}.

\section{Geometry Based DL Transmission}\label{sec:pc}
Observe that several implicit assumptions are made in the derivation of \eqref{eq:rate_opt}. In particular, perfect instantaneous \ac{csi} is required at all involved communication nodes and the beamforming matrices are computed centralized. Moreover, the per-satellite power constraints \eqref{eq:power_constr} are not necessarily met as a relaxed sum power constraint over all satellites is considered in \eqref{eq:rate_p2p}.
Obtaining the necessary \ac{csi} requires accurate channel estimation at the satellites, which then has to be shared with all other satellites in the swarm via inter-satellite links. Consequently, especially in the \ac{leo}, where the coherence time of the channel is very short, \ac{svd} based optimal precoding is not feasible.

However, the channels from the antenna elements of a single satellite to the \ac{rx} are highly correlated \cite{Schwarz.Delamotte.Storek.Knoop2019} and thus, although the matrix $\vekV\in\C^{\NTx\times\NTx}$ has $\NTx$ columns, there are only $M\le \NS\le \NTx$ singular values significantly larger than zero. 
Accordingly, only the right singular vectors corresponding to the $M$ largest singular values are of interest for designing the precoding matrix $\PC$.
Furthermore, the communication between satellites and \acp{rx} is usually done under \ac{los}. Thus, the channel matrix $\vekH$ is fully determined by the distances $d^\ell_{m,n}$ between transmit and receive antennas as well as atmospheric effects \cite{3GPP.TR.38.811,Storek.Hofmann.Knopp.2015}.

In this section, we exploit readily available position information to estimate the dominant large-scale components of $\vekH$. Based on this geometrical channel model, we design a distributed linear precoder that does not require any inter-satellite coordination and a linear equalization scheme at the \ac{rx} that does not rely on traditional \ac{csi} acquisition.

\subsection{Geometrical Channel Approximation} \label{sec:cm}
Due to the large distance between satellite $\ell$ and the \ac{rx} in relation to the antenna spacing, the \ac{aoa} and \ac{aod} between antennas in the transmit and receive arrays are approximately equal.
Thus, the distance $d^\ell_{m,n}$ from the $\ell$th satellite's $n$th antenna to the $m$th \ac{rx} antenna is approximately
\begin{align}
		d_{m,n}^\ell \approx d_\ell &- \DA(m-1)\cos(\aoal)
		- \DA(n-1)\sin(\aodl)    
\end{align}
where $d_\ell$ is the distance from the first transmit antenna at satellite $\ell$ to the first receive antenna. Moreover, the $m n$ channels from satellite $\ell$ to the ground station are subject to the same atmospheric effects \cite{Storek.Hofmann.Knopp.2015}. Thus, it is reasonable to assume that the entries in $\vekH_\ell$ have equal magnitude and differ only in their phase. In particular, let $\nu = 2\pi \fc/c_0$ be the wavenumber of the radiated carrier signal. Then, the phase difference between channels from adjacent transmit antennas to the same receive antenna is $\nu\DA\sin\left(\aodl\right) = \pi \sin\left(\aodl\right)$. Likewise, the phase difference between channels from a single transmit antenna to adjacent receive antennas is $\pi\cos\left(\aoal\right)$ \cite{You.Li.Wang.Gao.Xia.Ottersten.2020,Lin.Gao.Jin.Li.2017}. Thus, the $(m,n)$th entry of $\vekH_\ell$ is approximately $\alpha_\ell e^{j\pi\left((m-1)\cos(\aoal) +  (n-1)\sin(\aodl)\right)}$,
where $\alpha_\ell$ is the i.\,i.\,d.\ complex channel gain from satellite $\ell$ to the \ac{rx} with $\Exvl{\alpha_\ell}=0$ and $\var{\alpha_\ell}=\sigma_{\alpha}^2$.
As the satellites are following the same trajectory, the statistics of $\{\alpha_\ell\}_{\ell=1}^{\NS}$ are assumed to be the same for all satellites.

Define the steering vectors
\begin{equation}
	\veka_\ell^T = \left[ e^{j\pi m \cos(\aoal)} \right]_{m = 0}^{\Nr - 1},
	\quad
	\vekb_\ell^T = \left[ e^{-j\pi n \sin(\aodl)} \right]_{n = 0}^{\Nt - 1}.
\end{equation}
Then, the approximated channel matrix is 
\begin{align}\label{eq:ch_appr_sat}
	\chest_\ell = \alpha_\ell\veka_\ell\vekb_\ell^H \approx \chtx{\ell}
\end{align}
and has rank one. Due to $\Nr\ge\NS$ and the satellites having distinct positions in the orbital plane, i.e., $\aoa_i\neq\aoa_\ell$, for all $i\neq\ell$, the overall channel matrix $\chest=[\alpha_1\veka_1\vekb_1^H,...,\alpha_{\NS}\veka_{\NS}\vekb_{\NS}^H]$ has rank $\NS$. This allows for the parallel transmission of $M = \NS$ independent streams.

In the following, the precoder and equalizer are designed based on $\chest$. This only requires knowledge of the differential phases between the antennas that is straightforward to obtain from local position information, as shown above and in Section~\ref{sec:system}.

\subsection{Precoding}\label{sec:ad_pc}
Based on the observation, that we can transmit $M=\NS$ independent data streams in parallel, we propose the following  geometry based precoder
\begin{align}\label{eq:pc_ad}
	\PC_\text{geo} 
	&= \sqrt{\frac{1}{\Nt}} \blkdiag{\sqrt{\rho_1} \vekb_1, \dots, \sqrt{\rho_{\NS}} \vekb_{\NS}} \,.
\end{align}
Thus, satellite $\ell$ transmits into the direction of the eigenvector of $\Exvl{\chest^H\chest}$.
Let $\pc_{\ell,\text{geo}}=\sqrt{\rho_\ell/\Nt}\vekb_\ell$, then, due to the block diagonal precoding matrix, satellite $\ell$ transmits $\vekx_\ell =\pc_{\ell,\text{geo}} s_\ell$, i.e., it needs not know $\veks$ but only their part $s_\ell$ of the stream.
Note that according to \eqref{eq:pc_ad}, the per satellite average power constraint \eqref{eq:power_constr} is always satisfied.
Furthermore, satellite $\ell$ only has to know its \ac{aod} $\aodl$ and no cooperation between the satellites is needed to determine the precoding matrix $\PC_\text{geo}$.
In addition, the proposed precoding is based on manipulating only the phase at each antenna and thus, an efficient implementation with a single RF chain per satellite is possible \cite{Lin.Gao.Jin.Li.2017}.

\subsection{Linear Equalization}
In a satellite swarm, all satellites are usually of the same type \cite{Verhoeven.Bentum.Monna.Rotteveel.Guo.2011} and thus, it is assumed that all satellites transmit with the same power, in the following, i.e., $\rho_\ell=\rho$, for all $\ell$.
Assuming the previously proposed precoder \eqref{eq:pc_ad} and
employing a linear equalizer $\EQ=[\eq_1,...,\eq_{\NS}]^H\in\C^{\NS\times\Nr}$ at the \ac{rx}, the estimated signal is $\veksest=\EQ\vekH\PC_\text{geo}\veks + \EQ\vekn$.
Consequently, the signal transmitted by satellite $i$ interferes with the signal transmitted by satellite $\ell$. Then, the \ac{sinr} of the $\ell$th stream is
\begin{align}
	\Gamma_\ell &= \frac{\left\vert\eqsat^H\chtx{\ell}\pc_{\ell,\text{geo}}\right\vert^2} {\sum_{i\neq\ell} \left\vert \eqsat^H\chtx{i}\pc_{i,\text{geo}} \right\vert^2 + \nvar\eqsat^H\eqsat} \label{eq:sinr}\\
	&= \frac{\eqsat^H\chtx{\ell}\pc_{\ell,\text{geo}}\pc_{\ell,\text{geo}}^H\chtx{\ell}^H\eqsat} {\eqsat^H\left( \sum_{i\neq\ell} \chtx{i}\pc_{i,\text{geo}}\pc_{i,\text{geo}}^H\chtx{i}^H  + \nvar\vekI_{\Nr}\right)\eqsat^H} \label{eq:sinr_rayleigh_quotient}
\end{align}
and the achievable rate $R_\text{lin}$ is given by the sum of the individual rates
\begin{align}\label{eq:rate_sum}
	R_{\text{lin}} = \sum_{\ell=1}^{\NS}  \log_2\left(1 + \Gamma_\ell\right) 
\end{align}
Observe that $\Gamma_\ell$ is independent of $\eq_i$ for all $i\neq\ell$. Thus, \eqref{eq:rate_sum} is maximized by optimizing each $\Gamma_\ell$ separately. Since $\Gamma_\ell$ is a generalized Rayleigh quotient \cite{Horn1990}, it's maximizer is
\cite{Sadek.Tarighat.Sayed.2007}
\begin{align}\label{eq:eq_opt}
	\eq_{\ell,\text{opt}}^H  &= \pc_{\ell,\text{geo}}^H\chtx{\ell}^H \left(\sum_{i=1}^{\NS} \chtx{i}\pc_{i,\text{geo}}\pc_{i,\text{geo}}^H\chtx{i}^H + \nvar\vekI_{\Nr} \right)^{-1}.
\end{align}

However, acquiring perfect instantaneous \ac{csi} $\vekH$ is costly. Instead, we obtain the equalizer based on the approximated channel in Section~\ref{sec:cm}. Since $\chest_{\ell}\pc_{\ell,\text{geo}} = \alpha_\ell\sqrt{\Nt\rho}\veka_\ell$, the proposed equalizer is
\begin{subequations}\label{eq:eq_geo}
	\begin{align}
		\eq_{\ell,\text{geo}}^H  &= \veka_\ell^H \left(\sum_{i=1}^{\NS} \veka_i\veka_i^H + \naltvar\vekI_{\Nr} \right)^{-1} \\
		&= \veka_\ell^H \left(\vekA\vekA^H + \naltvar \vekI_{\Nr} \right)^{-1}
	\end{align}
\end{subequations}
where $\naltvar=\nvar/(\sigma_\alpha^2\Nt\rho)$ and $\vekA=[\veka_1,...,\veka_{\NS}]$.
Note that the proposed equalizer only requires the knowledge of the \acp{aoa} $\{\aoal\}_{\ell=1}^{\NS}$ from all satellites as well the \ac{snr} $1/\naltvar$ at the \ac{rx}.

\section{Optimal Inter-Satellite Distance}\label{sec:d_sat}
Based on the channel approximation \eqref{eq:ch_appr_sat}, the interconnection between the inter-satellite distance $\DS$ and the achievable rate is now analyzed. 
Assuming perfect \ac{csi} at the \ac{rx} and the fixed precoder $\PC_\text{geo}$, the ergodic rate $\tilde{R}$ for $\chest$ is upper bounded by
\begin{subequations}\label{eq:rate_approx}
	\begin{align}
		\tilde{R}_\text{opt} &\le \log_2\left\vert \vekI_{\Nr} + \frac{1}{\nvar} \Exvl{ \chest\PC_\text{geo}\PC_\text{geo}^H\chest^H} \right\vert \\
		&= \log_2\left\vert \vekI_{\Nr} + \frac{1}{\naltvar}\vekA\vekA^H \right\vert
	\end{align}
\end{subequations}
Thus, the achievable rate for $\chest$ is determined by the matrix $\vekA$, which is composed of the steering vectors $\{\veka_\ell\}_{\ell=1}^{\NS}$.

Due to the trail formation, the swarm is fully described by two parameters: The inter-satellite distance $\DS$ and the number of satellites $\NS$.
Choosing a proper inter-satellite distance $\DS$ is crucial, as it directly impacts the angular spread of the \acp{aoa} between the satellites,  which can be used to tune the matrix $\vekA$ such that the achievable rate is maximized, as stated in the following proposition.
\begin{proposition}
	The optimal inter-satellite distance w.r.t. the upper bound of the rate \eqref{eq:rate_approx} is achieved, if the following relation for the \ac{aoa} between every two satellites $\ell$ and $i$ holds
	\begin{align}\label{eq:orth_cond}
		\forall \ell\neq i: |\cos(\aoa_\ell) - \cos(\aoa_{i})| = \frac{2k}{\Nr}
	\end{align}
	where $k$ can be any positive integer number which is not a multiple of $\Nr$, i.e., $k$ must fulfill $\mod(k,\Nr)\neq 0$. 
\end{proposition}

\begin{proof}
	Observe that \eqref{eq:rate_approx} is equivalent to
	\begin{equation}\label{eq:rate_approx_ext}
		\tilde{R} \le \log_2\left\vert \vekI_{\NS} + \frac{1}{\naltvar}\vekA^H \vekA \right\vert
		= \log_2\left(\prod_{\ell=1}^{\NS}\left( 1 + \frac{\tilde{\lambda}_\ell}{\naltvar} \right)\right)
	\end{equation}
	where $\tilde{\lambda}_\ell$ are the positive eigenvalues of $\vekA^H \vekA$. Keeping the trace of $\vekA^H \vekA$ constant, this is maximized if all eigenvalues have the same value \cite[Thm.~2.21]{Jorswieck2007}. In other words, any $N_S\times N_S$ matrix $\vekB = \vekA^H \vekA$ maximizing \eqref{eq:rate_approx_ext} has a single eigenvalue $\lambda$ with multiplicity $N_S$.
	
	Further, observe that $\vekB$ is a normal matrix. By \cite[Thm.~2.5.4]{Horn1990}, $\vekB$ is similar matrix to a diagonal matrix, i.e., there exists a nonsingular matrix $\vekS$ such that $\vekS^{-1} \vek{\Lambda} \vekS = \vekB$ with $\vek{\Lambda}$ diagonal. Since similar matrices have the same eigenvalues \cite[Cor.~1.3.4]{Horn1990}, $\vek{\Lambda}$ must be $\lambda \vekI$. Then, for every nonsingular $\vekS$, we have $\vekB = \vekS^{-1} \lambda \vekI \vekS = \lambda \vekS^{-1} \vekS = \lambda \vekI$. It follows that $\vekB = \lambda\vekI$ is the unique maximizer of \eqref{eq:rate_approx_ext}.

	If the \acp{aoa} $\aoa_i$ and $\aoal$ of satellites $\ell$ and $i$ satisfy \eqref{eq:orth_cond}, the steering vectors $\veka_i$ and $\veka_\ell$ can be represented as different columns of the $\Nr\times\Nr$ \ac{dft} matrix and are thus orthogonal as the \ac{dft} is an orthogonal matrix, i.e.,
	\begin{subequations}\label{eq:orth_steering}
		\begin{align}
			\veka_i^H\veka_\ell &= \sum_{m=0}^{\Nr-1} e^{j\pi m\left( \cos(\aoa_\ell) - \cos(\aoa_{i}) \right)}\\
			&= \sum_{m=0}^{\Nr-1} e^{j2\pi \frac{km}{\Nr}} 
			= 0\,. 
		\end{align}
	\end{subequations}
	Consequently, the matrix $\vekA^H\vekA$ becomes a scaled identity matrix
	\begin{align}
		\vekA^H\vekA = \Nr\vekI_{\NS}
	\end{align}
	which maximizes the achievable rate \eqref{eq:rate_approx}, as all eigenvalues are identical.
\end{proof}

Consider now two neighbouring satellites $\ell$ and $\ell-1$.  
The difference of the cosine terms $\Delta\phi$ is given by 
\begin{subequations}\label{eq:DeltaPhi_gen}
	\begin{align}
		\Delta\phi &= \cos(\theta_{\ell}) - \cos(\theta_{\ell-1}) \\
		&=  \cos(\theta_{\ell}) - 
		\frac{r_0 \cos\left(\vartheta_\ell + \Delta\vartheta\right)}{\sqrt{d_0^2 + 2\RE r_0\left(1 -  \sin\left(\vartheta_\ell + \Delta\vartheta\right)\right)}} \,.
	\end{align}
\end{subequations}
where $\Delta\vartheta = \vartheta_{\ell-1} - \vartheta_\ell = \arccos\left(1 - \DS^2/(2r_0^2)\right)$ is the angular distance between both satellites $\ell-1$ and $\ell$, measured from the Earth's center, which is constant over time and identical for all neighbouring satellites.
In Fig. \ref{fig:DvsTheta}, the dependency between the required inter-satellite distance $D_\text{S,orth}$ and the \ac{aoa} $\aoal$ in degree to fulfill \eqref{eq:orth_cond} is shown for an altitude $d_0=\SI{600}{\kilo\meter}, k=1$ and different numbers of receive antennas $\Nr$.

Obviously, it is not possible to ensure orthogonal channels between all satellites during the whole flight with a constant inter-satellite distance $\DS$, as the \ac{aoa} changes over time. Adjusting the inter-satellite distance during the flight requires additional fuel and increased complexity for flight control and should thus be avoided.

However, as evaluated numerically in the next section, the capacity is not decreasing dramatically for $\Delta\Phi>2/\Nr$. Therefore, as a close to optimal heuristic, the condition \eqref{eq:orth_cond} can be relaxed, such that the average capacity over the whole flight is maximized if
\begin{align}\label{eq:orth_cond_relaxed}
	\min_\ell \Delta\phi = \min_\ell \cos(\aoa_\ell) - \cos(\aoa_{\ell-1}) \geq \frac{2}{\Nr}
\end{align}
holds for each time instance.

\begin{figure}[t]%
	\centering
\begin{tikzpicture}

\definecolor{color0}{rgb}{0.12156862745098,0.466666666666667,0.705882352941177}
\definecolor{color1}{rgb}{1,0.498039215686275,0.0549019607843137}
\definecolor{color2}{rgb}{0.172549019607843,0.627450980392157,0.172549019607843}

\begin{axis}[
width=\columnwidth,
height=2.5in,
legend cell align={left},
legend style={fill opacity=0.8, draw opacity=1, text opacity=1, draw=white!80!black},
log basis y={10},
xlabel={$\theta_\ell\;[^\circ]$},
xmajorgrids,
xmin=1, xmax=90,
ylabel={$D_{\text{S,orth}}\;[\text{km}]$},
ymajorgrids,
ymin=9.25435572039344, ymax=2000,
ymode=log,
]
\addplot [thick, color0, dashed]
table {%
0 2156.71089257438
1 2038.99033252771
2 1926.76839736645
3 1819.54521210101
4 1717.82065172098
5 1621.09484123677
6 1529.36778064839
7 1442.63946995583
8 1360.65997166431
9 1283.17934827902
10 1210.44753729477
11 1141.71472622719
12 1077.23085257105
13 1016.74597883157
14 959.760230019168
15 906.523543628636
16 856.536044670389
17 809.797733144429
18 765.808734061172
19 724.81898491541
20 686.32861071756
21 650.587548962414
22 616.845987165597
23 585.353862821902
24 555.861238436536
25 528.368114009501
26 502.624552046004
27 478.380615051254
28 455.636303025252
29 434.641553462789
30 414.64655387949
31 396.151179264939
32 378.655554629552
33 362.409617468122
34 346.913492791066
35 332.667055587966
36 318.920493374448
37 306.423618634886
38 294.426618884907
39 283.179431619302
40 272.68205683807
41 262.684557046421
42 253.436869739145
43 244.689057421452
44 236.441120093341
45 228.693057754813
46 221.194932911076
47 214.446620551713
48 207.948245687141
49 201.69980831736
50 195.951245937161
51 190.452621051754
52 185.45387115593
53 180.455121260105
54 175.956246353863
55 171.707308942412
56 167.708309025752
57 163.709309109092
58 160.210184182015
59 156.960996749729
60 153.711809317443
61 150.712559379948
62 147.963246937245
63 145.213934494541
64 142.714559546629
65 140.465122093508
66 138.215684640387
67 136.216184682057
68 134.466622218518
69 132.71705975498
70 130.967497291441
71 129.467872322694
72 128.218184848737
73 126.968497374781
74 125.718809900825
75 124.71905992166
76 123.969247437286
77 123.219434952913
78 122.469622468539
79 121.719809984165
80 121.469872489374
81 120.969997499792
82 120.720060005
83 120.470122510209
84 120.470122510209
85 120.470122510209
86 120.470122510209
87 120.720060005
88 121.219934994583
89 121.469872489374
90 122.219684973748
};
\addlegendentry{$N_\text{r}=10$}

\addplot [thick, color1, dashdotted]
table {%
0 1641.58971580965
1 1524.11909325777
2 1413.64672056005
3 1309.67272272689
4 1212.4470372531
5 1121.71972664389
6 1037.24085340445
7 958.760480040003
8 886.278606550546
9 819.045420451704
10 757.060921743479
11 700.075172931078
12 647.588299024919
13 599.350362530211
14 555.111425952163
15 514.371614301192
16 477.130927577298
17 442.889490790899
18 411.647303941995
19 382.904492041003
20 356.661055087924
21 332.417118093174
22 310.422618551546
23 290.177681473456
24 271.432369364114
25 254.43661971831
26 238.690557546462
27 224.194182848571
28 210.697558129844
29 198.450620885074
30 187.203433619468
31 176.706058838237
32 167.20843403617
33 158.210684223685
34 149.962746895575
35 142.214684557046
36 135.216434702892
37 128.71805983832
38 122.469622468539
39 116.721060088341
40 111.472372697725
41 106.723560296691
42 101.974747895658
43 97.725810484207
44 93.7268105675473
45 90.22768564047
46 86.7285607133928
47 83.4793732811068
48 80.2301858488207
49 77.4808734061172
50 74.7315609634136
51 72.4821235102925
52 69.9827485623802
53 67.7333111092591
54 65.7338111509292
55 63.7343111925994
56 61.9847487290608
57 60.2351862655221
58 58.7355612967747
59 57.2359363280273
60 55.7363113592799
61 54.4866238853238
62 53.2369364113676
63 51.9872489374115
64 50.9874989582465
65 49.9877489790816
66 48.9879989999167
67 48.238186515543
68 47.238436536378
69 46.4886240520043
70 45.7388115676306
71 44.9889990832569
72 44.4891240936745
73 43.989249104092
74 43.4893741145095
75 42.9894991249271
76 42.4896241353446
77 41.9897491457621
78 41.7398116509709
79 41.2399366613884
80 40.9899991665972
81 40.740061671806
82 40.4901241770147
83 40.2401866822235
84 40.2401866822235
85 39.9902491874323
86 39.9902491874323
87 39.9902491874323
88 39.9902491874323
89 39.9902491874323
90 39.9902491874323
};
\addlegendentry{$N_\text{r}=30$}

\addplot [thick, color2]
table {%
0 1095.7262271856
1 980.005167097258
2 875.03141928494
3 780.555046253854
4 696.076173014418
5 620.844987082257
6 554.361613467789
7 495.376364697058
8 443.639303275273
9 397.650804233686
10 357.160930077506
11 321.669805817151
12 290.177681473456
13 262.184682056838
14 237.690807567297
15 215.696308025669
16 196.451120926744
17 178.955496291358
18 163.709309109092
19 149.962746895575
20 137.465872156013
21 126.468622385199
22 116.721060088341
23 107.723310275856
24 99.7253104425369
25 92.4771230935911
26 85.7288107342278
27 79.9802483540295
28 74.4816234686224
29 69.4828735727977
30 65.2339361613468
31 60.9849987498958
32 57.2359363280273
33 53.9867488957413
34 50.7375614634553
35 47.9882490207517
36 45.2389365780482
37 42.7395616301358
38 40.4901241770147
39 38.4906242186849
40 36.7410617551463
41 34.9914992916076
42 33.241936828069
43 31.7423118593216
44 30.2426868905742
45 28.9929994166181
46 27.7433119426619
47 26.743561963497
48 25.743811984332
49 24.7440620051671
50 23.7443120260022
51 22.9944995416285
52 21.9947495624635
53 21.2449370780898
54 20.7450620885074
55 19.9952496041337
56 19.4953746145512
57 18.7455621301775
58 18.245687140595
59 17.7458121510126
60 17.2459371614301
61 16.7460621718477
62 16.4961246770564
63 15.996249687474
64 15.7463121926827
65 15.2464372031003
66 14.996499708309
67 14.7465622135178
68 14.4966247187266
69 14.2466872239353
70 13.9967497291441
71 13.7468122343529
72 13.4968747395616
73 13.2469372447704
74 13.2469372447704
75 12.9969997499792
76 12.7470622551879
77 12.7470622551879
78 12.4971247603967
79 12.4971247603967
80 12.2471872656055
81 12.2471872656055
82 12.2471872656055
83 11.9972497708142
84 11.9972497708142
85 11.9972497708142
86 11.9972497708142
87 11.9972497708142
88 11.9972497708142
89 11.9972497708142
90 11.9972497708142
};
\addlegendentry{$N_\text{r}=100$}
\end{axis}

\end{tikzpicture}
	\caption{Required inter-satellite distances $D_\text{S,orth}$ for different elevation angles $\theta_\ell$ and number of receive antennas $\Nr$ at altitude $d_0=600\,$km to achieve orthogonal steering vectors}
	\label{fig:DvsTheta}%
\end{figure}
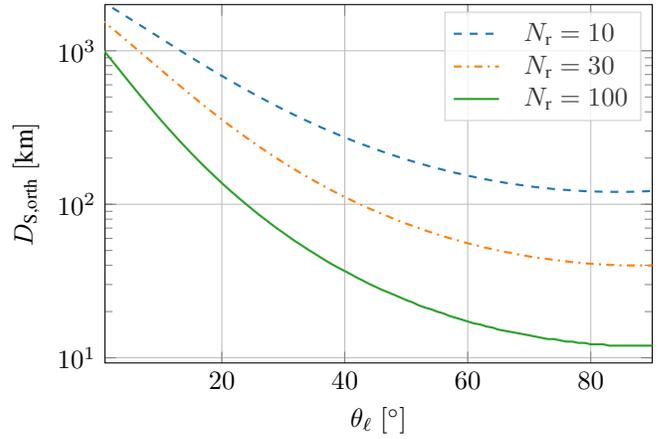

\section{Numerical Evaluations}\label{sec:simulation}
\subsection{Channel Model}\label{sec:channel}
In this paper, a pure \ac{los} based channel model between the \ac{rx} and each satellite is assumed and the carrier frequency is $\fc = \SI{20}{\giga\hertz}$.
The $(m,n)$th element of the channel matrix $\vekH_\ell$ is modeled by
\begin{align}\label{eq:channel_true}
	[\vekH_\ell]_{m,n}=h_{m,n}^\ell &=\frac{1}{\sqrt{L_{m,n}^\ell}}e^{-j\left(\nu d_{m,n}^\ell + \phi_{\text{atm},\ell}\right)}
\end{align}
where $\phi_{\text{atm},\ell}\in[0,2\pi]$ is a uniformly distributed phase shift caused by the atmosphere, and $L$ is the path loss, which is given in decibel as \cite{3GPP.TR.38.811}
\begin{align}
	\begin{split}
		L_{m,n|\dB}^\ell  =&\, 20\log_{10}\left(2\nu d_{m,n}\right) -\left(\GTx + \GRx \right)\\ 
		&+ L_{\text{sf},\ell} + L_{\text{cl},\ell} + L_{\text{gas},\ell} + L_{\text{ts},\ell}   
	\end{split}
\end{align}
where $L_{\text{sf},\ell}\sim\mathcal{N}(0,\sigma_{\text{sf},\ell}^2)$ and  $L_{\text{cl},\ell}$ are the shadow fading and clutter loss, repsectively. In LOS $L_{\text{cl},\ell}=\SI{0}{\dB}$ and $\sigma_{\text{sf},\ell}^2$ depends on the \ac{aoa}. The specific values can be found in \cite{3GPP.TR.38.811}, while in this paper the rural scenario has been considered.
$L_{\text{gas},\ell}$ includes atmospheric gas absorption described in \cite{ITUR.P.676} using the reference standard atmosphere \cite{ITUR.P.835}.
Eventually, $L_{\text{ts},\ell}$ includes the losses due tropospheric scintillation summarized in \cite{ITUR.P.618,ITUR.P.531} 
and $\GTx$ and $\GRx$ are the transmit and receive antenna gains, respectively.
The proposed precoding approach for satellite swarms is evaluated numerically in terms of the achievable rate. 
The sum transmit power $\Ptx$ of the satellite swarm and the total number of transmit antennas $\NTx=60$ is independent of the number of satellites $\NS$ inside the swarm, i.e., the power and transmit antennas of each satellite is $\rho=\Ptx/\NS$ and $\Nt=60/\NS$, respectively.
The \ac{rx} is assumed to be equipped with $\Nr=100$ receive antennas and the minimum elevation angle is  $\aoa_{\min}=30^\circ$. For better comparison, the transmission is assumed to start if the mean \ac{aoa} of all satellites $\aoa_{\text{mean}}=1/\NS\sum_{\ell=1}^{\NS}\aoal$ equals the minimum elevation angle $\aoa_{\min}$, i,e, the evaluation is done for $30^\circ\leq\aoa_{\text{mean}}\leq 150^\circ$.
Furthermore, the noise power is assumed as $P_{\text{N,dB}}=\SI{-120}{\dBW}$ and the the transmit and receive antenna gains are $G_{\text{Tx,dB}}=\SI{17.8}{\dBi}$ and $G_{\text{Rx,dB}}=\SI{20}{\dBi}$, respectively.
The altitude of the orbital plane is $d_0=\SI{600}{\kilo\meter}$ during all simulations.

\subsection{Inter-satellite distance}
In section \ref{sec:d_sat}, the optimal inter-satellite distance $D_{\text{S,orth}}$ based on the channel approximation \eqref{eq:ch_appr_sat} has been derived.
In Fig. \ref{fig:RvsDs_th=90}, the achievable rate \eqref{eq:rate_p2p} is shown in dependence of the inter-satellite distance $\DS$ for a total transmit power $\Ptx=\SI{10}{\dBW}$.
It can be observed that the rate increases with an increasing inter-satellite distance up to $\DS = \SI{12}{\kilo\meter}$. Afterwards, it periodically increases and decreases, slightly, independent on the number of satellites $\NS$. The local maxima are every $12\,$km, which corresponds to section \ref{sec:d_sat}.
In Fig. \ref{fig:RvsDs_avrg}, the achievable rate is averaged over the considered time. Then, there are no more periodic variations of the achievable rate, but the achievable rate does not further increase for $\DS\ge \SI{65}{\kilo\meter}$, which justifies the heuristic \eqref{eq:orth_cond_relaxed}.

\begin{figure}[t]%
	\centering
%
%
\definecolor{mycolor1}{rgb}{0.85000,0.32500,0.09800}%
\definecolor{darkgreen}{rgb}{0.12549019607843137255,0.4980392156862745098,0.16862745098039215686}
\begin{tikzpicture}

\begin{axis}[%
width=\columnwidth,
height=2.5in,
xmin=0,
xmax=40,
xlabel style={font=\color{white!15!black}},
xlabel={$\text{inter-satellite  distance D}_\text{S}\text{ [km]}$},
ymin=12,
ymax=26,
ylabel style={font=\color{white!15!black}},
ylabel={Achievable rate [bps/Hz]},
axis background/.style={fill=white},
xmajorgrids,
ymajorgrids,
legend style={fill opacity=0.8, draw opacity=1, text opacity=1, draw=white!80!black}
]

\addplot [color=orange, thick]
  table[row sep=crcr]{%
0.5	10.1255324155257\\
1	11.7479876496492\\
1.5	12.930857183591\\
2	13.9985596926964\\
2.5	15.0669207194055\\
3	16.1068174178278\\
3.5	17.0571876736074\\
4	17.9377894788624\\
4.5	18.7575835033297\\
5	19.5778925126823\\
5.5	20.3858946008194\\
6	21.1635319847964\\
6.5	21.9325699662937\\
7	22.6295239675046\\
7.5	23.260567524234\\
8	23.8054910329948\\
8.5	24.2797920056965\\
9	24.6662445659414\\
9.5	24.974368044301\\
10	25.220330643267\\
10.5	25.3906816091009\\
11	25.4997294111387\\
11.5	25.5643997176315\\
12	25.5765989548631\\
12.5	25.5537638141506\\
13	25.5408887150232\\
13.5	25.4957903948296\\
14	25.4582261454355\\
14.5	25.4300484039219\\
15	25.4049790631133\\
15.5	25.3894891941869\\
16	25.3676370093473\\
16.5	25.3667913603546\\
17	25.3530044712056\\
17.5	25.3699678485674\\
18	25.3689622553336\\
18.5	25.3758984319784\\
19	25.3835743676457\\
19.5	25.3839844347121\\
20	25.4069267007065\\
20.5	25.4329309025424\\
21	25.4508198597149\\
21.5	25.4839193606471\\
22	25.5031217930032\\
22.5	25.5375146081851\\
23	25.5440106291069\\
23.5	25.5638491590187\\
24	25.5628110577967\\
24.5	25.5721756428297\\
25	25.5552141623656\\
25.5	25.5465925806835\\
26	25.5309804913347\\
26.5	25.5159563188274\\
27	25.5146757893001\\
27.5	25.4984601596686\\
28	25.4894147952172\\
28.5	25.4925577634933\\
29	25.4975486440228\\
29.5	25.4933020660661\\
30	25.494465429245\\
30.5	25.486796651859\\
31	25.4981293149472\\
31.5	25.4841951822782\\
32	25.494362216211\\
32.5	25.5141016039127\\
33	25.5119148929773\\
33.5	25.523471796332\\
34	25.5188257778783\\
34.5	25.5392958447767\\
35	25.5492542031451\\
35.5	25.5912401680942\\
36	25.5413933624005\\
36.5	25.5519129513548\\
37	25.5010384486047\\
37.5	25.5475780006073\\
38	25.5685389833594\\
38.5	25.5824399359265\\
39	25.5471414136513\\
39.5	25.5599500441954\\
40	25.468159057958\\
40.5	25.576225687274\\
41	25.4857582500946\\
41.5	25.5153925520714\\
42	25.5041045015712\\
42.5	25.5420953454189\\
43	25.5089127313156\\
43.5	25.499752613179\\
44	25.5039865143964\\
44.5	25.5327180073775\\
45	25.4790068249412\\
45.5	25.5803481534212\\
46	25.4940815451595\\
46.5	25.5369598641992\\
47	25.500274218782\\
47.5	25.5462026742197\\
48	25.5623594923714\\
48.5	25.5345726826059\\
49	25.6223644054117\\
49.5	25.5394391288427\\
50	25.5107705964129\\
};
\addlegendentry{$N_\text{S}=4$}

\addplot [color=cyan, thick]
  table[row sep=crcr]{%
0.5	10.2217622406192\\
1	11.7189944152733\\
1.5	12.8081284470951\\
2	13.6689588375534\\
2.5	14.4733139646876\\
3	15.287146070729\\
3.5	16.0895081211618\\
4	16.8474554982597\\
4.5	17.5665376345699\\
5	18.2190313710134\\
5.5	18.8038532653007\\
6	19.3360944269667\\
6.5	19.7715554464557\\
7	20.1741619484878\\
7.5	20.4959832071532\\
8	20.8000147895885\\
8.5	21.0190188894227\\
9	21.2095771012354\\
9.5	21.3509451126785\\
10	21.4721484600366\\
10.5	21.5463172935298\\
11	21.607304750088\\
11.5	21.6363614038941\\
12	21.6433861738493\\
12.5	21.654517354463\\
13	21.6238279354005\\
13.5	21.5991439832203\\
14	21.5793813923505\\
14.5	21.5534986393553\\
15	21.5318998415956\\
15.5	21.5216619666365\\
16	21.5325084270524\\
16.5	21.5101513569546\\
17	21.510279525671\\
17.5	21.5080815649704\\
18	21.5063694942358\\
18.5	21.5166398027482\\
19	21.5290974535158\\
19.5	21.5349608989058\\
20	21.5643097033102\\
20.5	21.570270812878\\
21	21.5622824273548\\
21.5	21.5917211704786\\
22	21.6126447628762\\
22.5	21.6208193064255\\
23	21.6232394137292\\
23.5	21.6346237494846\\
24	21.649534937578\\
24.5	21.6411074947933\\
25	21.6503524581304\\
25.5	21.6404243436098\\
26	21.6132081851529\\
26.5	21.6201637831017\\
27	21.6256707037923\\
27.5	21.608209576596\\
28	21.6003311682199\\
28.5	21.603610375605\\
29	21.586376855776\\
29.5	21.5815976873282\\
30	21.5969496633806\\
30.5	21.5930810550359\\
31	21.6036540367258\\
31.5	21.6145008996882\\
32	21.5967910036604\\
32.5	21.5999797724608\\
33	21.596871165585\\
33.5	21.6215002959861\\
34	21.628805205119\\
34.5	21.6437031973035\\
35	21.6295735950396\\
35.5	21.6454184615643\\
36	21.6369077762289\\
36.5	21.6253495635218\\
37	21.6311259681052\\
37.5	21.6317105554804\\
38	21.6245701005094\\
38.5	21.6209163437352\\
39	21.6388458071101\\
39.5	21.6206251453307\\
40	21.6163789253118\\
40.5	21.610707302935\\
41	21.6228824960697\\
41.5	21.6105972443229\\
42	21.629591211721\\
42.5	21.6107482518488\\
43	21.6197123311867\\
43.5	21.6156443753618\\
44	21.6165996031217\\
44.5	21.6315275590422\\
45	21.6274912048089\\
45.5	21.6171911032522\\
46	21.6318306924336\\
46.5	21.6295802854971\\
47	21.6264226470153\\
47.5	21.6270678006615\\
48	21.6253328761082\\
48.5	21.6255501417683\\
49	21.6368239996727\\
49.5	21.6265014191687\\
50	21.6207740842157\\
};
\addlegendentry{$N_\text{S}=3$}

\addplot [color=darkgreen, thick]
  table[row sep=crcr]{%
0.5	10.339100635953\\
1	11.6290451340289\\
1.5	12.6172681943916\\
2	13.3445035325958\\
2.5	13.9297424050884\\
3	14.3908591146227\\
3.5	14.7903125897915\\
4	15.1218541896026\\
4.5	15.4002878043888\\
5	15.6392839543662\\
5.5	15.8446325115363\\
6	16.0196838087294\\
6.5	16.1674458834088\\
7	16.2934486828024\\
7.5	16.4028426992989\\
8	16.4928235018268\\
8.5	16.5720413483767\\
9	16.618164118217\\
9.5	16.6733528300013\\
10	16.7124498935479\\
10.5	16.7319685942617\\
11	16.7552816053409\\
11.5	16.7592694896485\\
12	16.7594455559335\\
12.5	16.7536723949177\\
13	16.7513913707899\\
13.5	16.7442359388605\\
14	16.7322832290527\\
14.5	16.7290292803429\\
15	16.7167342455471\\
15.5	16.7017634737712\\
16	16.7009332043698\\
16.5	16.6907235907057\\
17	16.6942741437706\\
17.5	16.6929983426796\\
18	16.6952429141111\\
18.5	16.7026259754324\\
19	16.7060099460915\\
19.5	16.7248114209639\\
20	16.7175502114775\\
20.5	16.7265577191882\\
21	16.7423797159438\\
21.5	16.7437637168174\\
22	16.7508786137676\\
22.5	16.7536950263092\\
23	16.7528535538587\\
23.5	16.7556452696712\\
24	16.7613266044532\\
24.5	16.7572064198726\\
25	16.7628612371623\\
25.5	16.7622488501648\\
26	16.7555048720723\\
26.5	16.7550999444397\\
27	16.7466826002583\\
27.5	16.7472318812829\\
28	16.7443811080702\\
28.5	16.7407649705872\\
29	16.7342368265962\\
29.5	16.740462426488\\
30	16.7315843737969\\
30.5	16.7359400911868\\
31	16.73785365127\\
31.5	16.7451562114257\\
32	16.7458141329027\\
32.5	16.7427191055509\\
33	16.7469921540389\\
33.5	16.754793776888\\
34	16.7564107589979\\
34.5	16.7666383353348\\
35	16.7596392291087\\
35.5	16.7570635174865\\
36	16.7660218073633\\
36.5	16.7593569625862\\
37	16.7659270000331\\
37.5	16.7485397764405\\
38	16.7558719647734\\
38.5	16.7488792609282\\
39	16.758059242809\\
39.5	16.7454554825319\\
40	16.7473334271065\\
40.5	16.7502949217953\\
41	16.7542820106216\\
41.5	16.7405096297983\\
42	16.7462311835673\\
42.5	16.7471463981691\\
43	16.7453812416053\\
43.5	16.7504672566872\\
44	16.7515656647148\\
44.5	16.7457597276296\\
45	16.7499959834071\\
45.5	16.7573915320393\\
46	16.7552215013674\\
46.5	16.7599106355306\\
47	16.7611451541791\\
47.5	16.7603623138549\\
48	16.76203406035\\
48.5	16.7594457770923\\
49	16.7544798309685\\
49.5	16.7617261522019\\
50	16.7528986619922\\
};
\addlegendentry{$N_\text{S}=2$}

\end{axis}

\end{tikzpicture}%
	\caption{Achievable rate performance for different inter-satellite distances $\DS$ and number of satellites $\NS$ for fixed \ac{aoa} $\aoa_{\text{mean}}=90^\circ$}
	\label{fig:RvsDs_th=90}%
\end{figure}
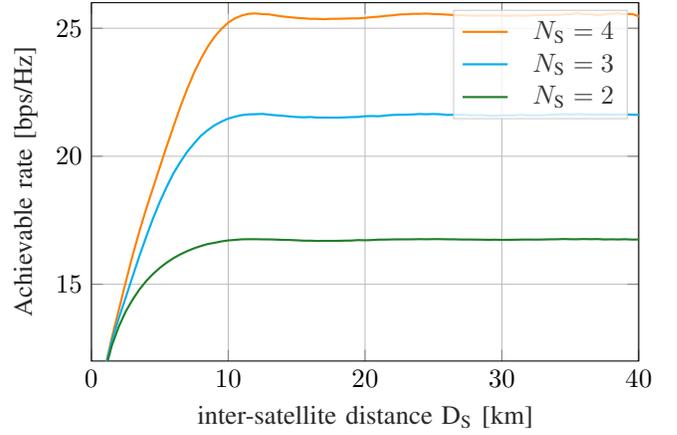

\begin{figure}[t]%
	\centering
%
%
\definecolor{darkgreen}{rgb}{0.12549019607843137255,0.4980392156862745098,0.16862745098039215686}
\begin{tikzpicture}

\begin{axis}[%
width=\columnwidth,
height=2.5in,
xmin=0,
xmax=100,
xlabel style={font=\color{white!15!black}},
xlabel={$\text{inter-satellite  distance D}_\text{S}\text{ [km]}$},
ymin=8,
ymax=28,
ylabel style={font=\color{white!15!black}},
ylabel={Achievable rate [bps/Hz]},
axis background/.style={fill=white},
xmajorgrids,
ymajorgrids,
legend style={fill opacity=0.8, draw opacity=1, text opacity=1, draw=white!80!black}
]
\addplot [color=purple, thick]
  table[row sep=crcr]{%
1	9.1757910721403\\
2	10.8434558061245\\
4	13.5852774111313\\
6	15.9324502346965\\
8	18.0047843913604\\
10	19.8520927732398\\
12	21.2485479402961\\
14	22.2146443399616\\
16	22.9221176050389\\
18	23.5069344378319\\
20	23.9974673897107\\
22	24.4280072691305\\
24	24.7694630804564\\
26	25.0958575564021\\
28	25.3446935610984\\
30	25.5602267487694\\
32	25.7824390772427\\
34	25.9224025731433\\
36	26.0753027198851\\
38	26.2055399098038\\
40	26.3082178730492\\
42	26.3853207848593\\
44	26.4757298415678\\
46	26.5447970293415\\
48	26.6023755069169\\
50	26.6704682009968\\
52	26.6820644267898\\
54	26.7280384807246\\
56	26.7645735881642\\
58	26.7640454672549\\
60	26.7941400655831\\
62	26.8184112709729\\
64	26.8186278045743\\
66	26.8523149182785\\
68	26.8301756353081\\
70	26.8249031234747\\
72	26.8295294237621\\
74	26.8363119916655\\
76	26.8261152054946\\
78	26.8374723899744\\
80	26.8272093853865\\
82	26.8026199240188\\
84	26.8423517085543\\
86	26.8331952166591\\
88	26.8048746441245\\
90	26.8127310376683\\
92	26.8338430610877\\
94	26.8151075186333\\
96	26.7952919286061\\
98	26.796327812666\\
100	26.7739212620196\\
};
\addlegendentry{$N_\text{S}=6$}

\addplot [color=teal, thick]
  table[row sep=crcr]{%
1	9.25397818816901\\
2	10.7487413937109\\
4	13.2410983219253\\
6	15.3541009398416\\
8	17.2334078537882\\
10	18.8481205946544\\
12	20.0408096735948\\
14	20.8503702312513\\
16	21.4815531462882\\
18	21.9719369928884\\
20	22.4093983989621\\
22	22.7829323833283\\
24	23.1152187770341\\
26	23.3723112593191\\
28	23.5709679300419\\
30	23.7786211812336\\
32	23.9457722378539\\
34	24.0878132320883\\
36	24.2476310847852\\
38	24.3388606844136\\
40	24.4648058287679\\
42	24.5317454810698\\
44	24.6133180454451\\
46	24.6665902822859\\
48	24.7121524182496\\
50	24.7577303424236\\
52	24.790922540686\\
54	24.8143323913368\\
56	24.8395322562462\\
58	24.8676079098669\\
60	24.8889062030565\\
62	24.8974002812506\\
64	24.888114979272\\
66	24.9083613522706\\
68	24.8839098729681\\
70	24.890786391699\\
72	24.9037859057941\\
74	24.9010654270077\\
76	24.8820614356924\\
78	24.9257880669984\\
80	24.9045555813914\\
82	24.8952418831203\\
84	24.9073709633728\\
86	24.8853212376605\\
88	24.9017767614686\\
90	24.8777090868902\\
92	24.9003497677209\\
94	24.8819596983581\\
96	24.8794497315834\\
98	24.8996333286357\\
100	24.8727801684399\\
};
\addlegendentry{$N_\text{S}=5$}

\addplot [color=orange, thick]
  table[row sep=crcr]{%
1	9.33082935771423\\
2	10.6605921385212\\
4	12.8884119445641\\
6	14.7005535528788\\
8	16.3310700203207\\
10	17.6009729161619\\
12	18.5338292082408\\
14	19.1920113951015\\
16	19.7087320735269\\
18	20.1009556975162\\
20	20.4566436357638\\
22	20.7501528575466\\
24	21.0171078453308\\
26	21.2330601886621\\
28	21.4130159138006\\
30	21.5641953702301\\
32	21.7026055589029\\
34	21.8224202923337\\
36	21.9362889563251\\
38	22.0268275835884\\
40	22.0824527849929\\
42	22.154440228352\\
44	22.2312334898042\\
46	22.2615446638538\\
48	22.305006806172\\
50	22.358994565254\\
52	22.3623752374677\\
54	22.3832149862159\\
56	22.4032168165001\\
58	22.4135290727763\\
60	22.4127549519952\\
62	22.4293235053284\\
64	22.4536727515465\\
66	22.4326012004344\\
68	22.4326748622245\\
70	22.4564804631717\\
72	22.4338568331074\\
74	22.4286269758764\\
76	22.4036721576734\\
78	22.4517513833528\\
80	22.4233529425608\\
82	22.4269158983825\\
84	22.4380471189684\\
86	22.4350121997379\\
88	22.4128705846402\\
90	22.4464969302113\\
92	22.4278398065121\\
94	22.440599577328\\
96	22.4241073993956\\
98	22.43650972949\\
100	22.4146282676278\\
};
\addlegendentry{$N_\text{S}=4$}

\addplot [color=cyan, thick]
  table[row sep=crcr]{%
1	9.4500035915883\\
3	11.5617733035724\\
5	13.2431049791139\\
7	14.5903196415186\\
9	15.6079265807424\\
11	16.3525536886203\\
13	16.9195114018191\\
15	17.2897998371028\\
17	17.6110278108739\\
19	17.8763495075592\\
21	18.0989311255154\\
23	18.2994266317655\\
25	18.4662460804495\\
27	18.5844233681079\\
29	18.7016066174825\\
31	18.8242124421497\\
33	18.890749850476\\
35	18.9606195550306\\
37	19.0181679759818\\
39	19.0731081975021\\
41	19.1108626956514\\
43	19.1471500720614\\
45	19.1807794158175\\
47	19.2032503928163\\
49	19.233293183489\\
51	19.2331961717542\\
53	19.2499986051969\\
55	19.2696621113669\\
57	19.2700832792217\\
59	19.2767938397313\\
61	19.2781531261634\\
63	19.2912639667852\\
65	19.3039304557692\\
67	19.2921105519952\\
69	19.3083060404941\\
71	19.3007953513656\\
73	19.3033368693381\\
75	19.2726190430399\\
77	19.2876923620712\\
79	19.2960109582022\\
81	19.2767511258473\\
83	19.2833541306525\\
85	19.2890304859413\\
87	19.2842107902383\\
89	19.2926402285412\\
91	19.2866513038923\\
93	19.3040562374582\\
95	19.293446512622\\
97	19.2835319773911\\
99	19.2909867963272\\
101	19.2900053023082\\
};
\addlegendentry{$N_\text{S}=3$}

\addplot [color=darkgreen, thick]
  table[row sep=crcr]{%
1	9.59829055332669\\
2	10.5354555799665\\
3	11.2797518643736\\
4	11.8677821465109\\
5	12.3571961452281\\
6	12.746359195249\\
7	13.0698953746112\\
8	13.3654084774593\\
9	13.5781731476237\\
10	13.7865492241589\\
11	13.9457624346968\\
12	14.0782798704673\\
13	14.1919899697558\\
14	14.3083659566549\\
15	14.3849933607014\\
16	14.4784027673185\\
17	14.5352650692162\\
18	14.5779688477877\\
19	14.6606618705045\\
20	14.7032715874655\\
21	14.7462389900966\\
22	14.7698586066029\\
23	14.8294629462552\\
24	14.8666343963085\\
25	14.8896268709689\\
26	14.9187659846287\\
27	14.9610883067561\\
28	14.9687560878954\\
29	14.9746194649883\\
30	15.004702757618\\
31	15.0184436422138\\
32	15.041589584657\\
33	15.0530914185252\\
34	15.0722703498901\\
35	15.0733224628022\\
36	15.0826595543562\\
37	15.0905693525758\\
38	15.1068325899779\\
39	15.1353647781318\\
40	15.1159533912387\\
41	15.1278482063509\\
42	15.1488957398995\\
43	15.1420413374575\\
44	15.1418273998224\\
45	15.1485386760993\\
46	15.1585325492172\\
47	15.1625020553414\\
48	15.1682945425688\\
49	15.1689827149399\\
50	15.1736899987045\\
51	15.1624808998198\\
52	15.1716110979493\\
53	15.162958453511\\
54	15.1772412309794\\
55	15.1919446308298\\
56	15.1796904596913\\
57	15.1815836580332\\
58	15.1899220198914\\
59	15.1761295795782\\
60	15.183497315917\\
61	15.1824593296486\\
62	15.1900226315118\\
63	15.1938731524763\\
64	15.1886410967875\\
65	15.1843978556323\\
66	15.1760644000527\\
67	15.1864107735152\\
68	15.1884851083161\\
69	15.1618535232411\\
70	15.1967942152448\\
71	15.1844018694884\\
72	15.181997420853\\
73	15.186596692529\\
74	15.1983342199294\\
75	15.1895104247826\\
76	15.1902389957672\\
77	15.1820422982842\\
78	15.1671096226836\\
79	15.1842831436322\\
80	15.1816791847825\\
81	15.1840254937785\\
82	15.1868857886663\\
83	15.1842433760061\\
84	15.185390062297\\
85	15.1865915352414\\
86	15.187956950338\\
87	15.1824975997289\\
88	15.1888438693115\\
89	15.1790067091223\\
90	15.1810881417464\\
91	15.1880944759037\\
92	15.1715018332287\\
93	15.1792309334284\\
94	15.1929528643536\\
95	15.1770788186823\\
96	15.1815428030111\\
97	15.1804298218343\\
98	15.187990134716\\
99	15.1821525652448\\
100	15.1811134373587\\
};
\addlegendentry{$N_\text{S}=2$}

\addplot [color=black, thick]
  table[row sep=crcr]{%
1	9.16559237386733\\
5	9.18730569786663\\
10	9.18486698051549\\
15	9.16873475960074\\
20	9.18920141011328\\
25	9.16819315053118\\
30	9.17485588821525\\
35	9.16875418368924\\
40	9.17730150140786\\
45	9.18372085461609\\
50	9.18037296468834\\
55	9.1733637205549\\
60	9.16715245764252\\
65	9.1733220950467\\
70	9.18139500958191\\
75	9.18347543740666\\
80	9.18667173518371\\
85	9.17855293910397\\
90	9.15973961403272\\
95	9.17726953401648\\
100	9.17066248715768\\
};
\addlegendentry{$N_\text{S}=1$}

\end{axis}

\end{tikzpicture}%
	\caption{Achievable rate performance for different inter-satellite distances $\DS$ and number of satellites $\NS$ averaged over time}
	\label{fig:RvsDs_avrg}%
\end{figure}
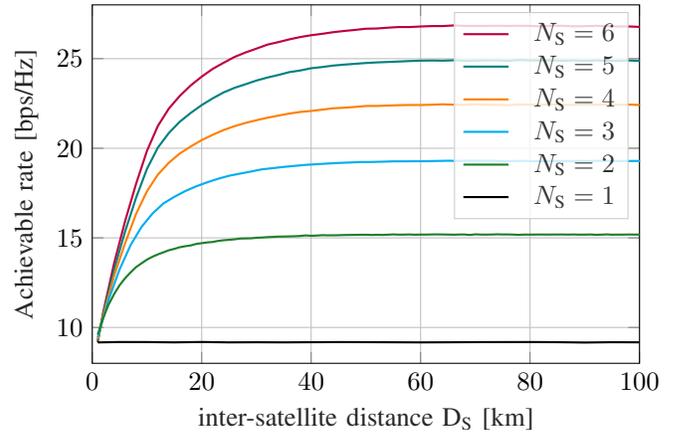

\subsection{Geometry Based DL}
Throughout this paper, two different precoder and two different equalizer approaches have been discussed. In Fig. \ref{fig:RvsSNR_SVD} three possible combinations are compared.

First, the optimal precoder is given by the \ac{svd} \eqref{eq:pc_opt}. Further, assume an ideal, and in general non-linear, equalizer, the achievable rate $R_{\text{opt}}$ is given by \eqref{eq:rate_opt}.
Second, consider the geometry based precoder \eqref{eq:pc_ad} and again, an ideal equalizer. Then, the achievable rate is given by
\begin{align}
	R_\text{per} = \log_2\left\vert \vekI_{\Nr} + \frac{1}{\sigma_\mathsf{n}^2} \vekH\PC_{\text{geo}}\PC_{\text{geo}}^H\vekH^H \right\vert\,.
\end{align}

Finally, consider again the geometric precoder \eqref{eq:pc_ad} and the linear equalizer \eqref{eq:eq_geo}. The corresponding achievable rate is given by $R_{\text{lin}}$ in \eqref{eq:rate_sum}.
In Fig. \ref{fig:RvsSNR_SVD} these three approaches are compared. If the inter-satellite distance is $\DS=\SI{70}{\kilo\meter}$, i.e., $\DS>D_{\text{S,orth}}=\SI{65}{\kilo\meter}$ for $\aoa_{\text{mean}}=30^\circ$, all three approaches achieve almost the same performance. 
If on the other hand $\DS=\SI{10}{\kilo\meter}$, the achievable rate decreases with a linear equalizer. However, with the proposed precoder \eqref{eq:pc_ad}, still, the optimum achivable rate can be reached.

\begin{figure}[t]%
	\centering
%
%
\definecolor{mycolor1}{rgb}{0.85000,0.32500,0.09800}%
\definecolor{mycolor2}{rgb}{0.00000,0.44700,0.74100}%
\begin{tikzpicture}

\begin{axis}[%
width=\columnwidth,
height=2.5in,
xmin=0,
xmax=40,
xlabel style={font=\color{white!15!black}},
xlabel={Total transmit power $\text{P}_{\text{Tx}}\text{ in dBW}$},
ymin=5,
ymax=50,
ylabel style={font=\color{white!15!black}},
ylabel={Achievable rate [bps/Hz]},
axis background/.style={fill=white},
xmajorgrids,
ymajorgrids,
legend style={at={(0.01,0.98)}, anchor=north west, fill opacity=0.8}
]

\addplot [color=purple, dashed, thick, mark=triangle, mark options={solid}]
  table[row sep=crcr]{%
0	9.85380367151781\\
10	19.2847515592654\\
20	29.1911554612231\\
30	39.1509299947631\\
40	49.1161125701864\\
};
\addlegendentry{$R_{\text{opt}}; D_\text{S}=70\,\text{km}$}

\addplot [color=cyan, dashdotted, thick, mark=triangle, mark options={solid}]
  table[row sep=crcr]{%
0	9.85684330437001\\
10	19.2870403969395\\
20	29.1933521231214\\
30	39.1531173622396\\
40	49.1182990077574\\
};
\addlegendentry{$R_{\text{per}}; D_\text{S}=70\,\text{km}$}

\addplot [color=orange, dotted, thick, mark=triangle, mark options={solid}]
  table[row sep=crcr]{%
0	9.83196711981662\\
10	19.2545809048704\\
20	29.1601897452807\\
30	39.1199155607385\\
40	49.0850887202012\\
};
\addlegendentry{$R_{\text{lin}}; D_\text{S}=70\,\text{km}$}

\addplot [color=purple, dashed, thick, mark=triangle, mark options={solid, rotate=180}]
  table[row sep=crcr]{%
0	8.18698310904728\\
10	16.0136077908772\\
20	24.9643852851921\\
30	34.5636810274722\\
40	44.461029803624\\
};
\addlegendentry{$R_{\text{opt}}; D_\text{S}=10\,\text{km}$}

\addplot [color=cyan, dashdotted, thick, mark=triangle, mark options={solid, rotate=180}]
  table[row sep=crcr]{%
0	8.19483494025138\\
10	16.0244324154294\\
20	24.9772287177199\\
30	34.576791553212\\
40	44.4740089020305\\
};
\addlegendentry{$R_{\text{per}}; D_\text{S}=10\,\text{km}$}

\addplot [color=orange, dotted, thick, mark=triangle, mark options={solid, rotate=180}]
  table[row sep=crcr]{%
0	6.53802176056534\\
10	12.3199836529525\\
20	19.5616104982513\\
30	28.7348745087126\\
40	38.6017602226161\\
};
\addlegendentry{$R_{\text{lin}}; D_\text{S}=10\,\text{km}$}

\end{axis}
\end{tikzpicture}%
	\caption{Achievable rate performance for different transmit powers $\Ptx$ and inter-satellite distances $\DS$. Number of satellites is $\NS=3$}
	\label{fig:RvsSNR_SVD}%
\end{figure}
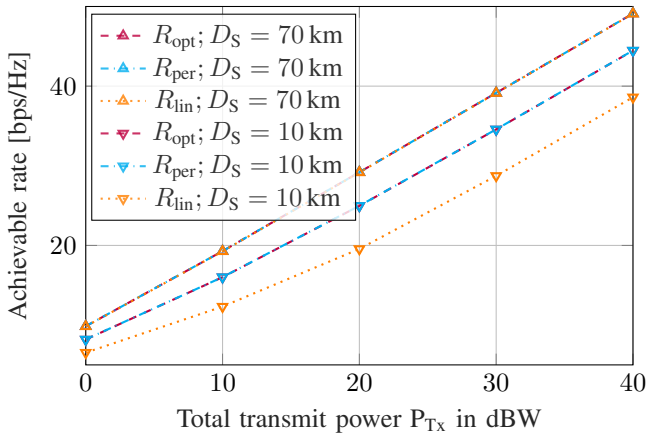

\balance

\section{Discussion}\label{sec:conclusion}
In this paper, we developed a low complexity distributed precoder for satellite swarms and a linear equalizer, both utilizing only the geometric relation between the satellites and \ac{rx} positions.
Given that the inter-satellite distances are chosen adequately, we have shown that the proposed precoder-equalizer combination achieves a performance very close to the capacity upper bound obtained by assuming perfect CSI and instantaneous coordination between satellites. This only requires positional knowledge at all terminals, a small amount of CSI at the \ac{rx}, i.e., tracking of one scalar channel coefficient per satellite, no CSI at the transmitter and no active coordination between satellites.
Of course, in a real world system even these assumptions might not hold. In particular,
the satellite positions are subject to small perturbations and channel coefficients are difficult to track perfectly in this high mobility scenario. However, these aspects can be incorporated in the system design to make it robust against such imperfections, as we will show in the journal extension of this paper.

\section*{Acknowledgment}
This research was supported in part by the German Federal Ministry of Education and Research (BMBF) within the project Open6GHub under grant number 16KISK016A and 
by the German Research Foundation (DFG) under Germany's Excellence Strategy (EXC 2077 at University of Bremen, University Allowance).

\bibliographystyle{./my_lib/IEEEtran}
\bibliography{./my_lib/IEEEabrv,./my_lib/dip_references}

\end{document}